\documentclass{article}

\usepackage{PRIMEarxiv}

\usepackage[utf8]{inputenc} 
\usepackage[T1]{fontenc}    
\usepackage{hyperref}       
\usepackage{url}            
\usepackage{booktabs}       
\usepackage{amsfonts}       
\usepackage{nicefrac}       
\usepackage{microtype}      
\usepackage{lipsum}
\usepackage{fancyhdr}       
\usepackage{graphicx}       
\graphicspath{{media/}}     
\usepackage{amsmath}
\usepackage{amssymb}
\usepackage{amsthm}
\usepackage{comment}
\usepackage{xcolor}
\usepackage{bbm}

\theoremstyle{definition}
\newtheorem{definition}{Definition}[section]

\newtheorem{remark}{Remark}
\newtheorem{theorem}{Theorem}

\graphicspath{{figures/}}

\title{Intrinsic Random Functions and Parametric Covariance Models of Spatio-Temporal Random Processes on the Sphere
}

\author{
  Jongwook Kim \\
  Indiana University\\
  Bloomington \\
  \texttt{jki5@iu.edu} \\
  \And
    Chunfeng Huang \\
  Indiana University\\
  Bloomington \\
  \texttt{huang48@iu.edu} \\
   \And
  Nicholas Bussberg \\
  Elon University \\
  Elon\\
  \texttt{nbussberg@elon.edu} \\
}

\begin{document}
\maketitle

\begin{abstract}
Identifying an appropriate covariance function is one of the primary interests in spatial and spatio-temporal statistics because it allows researchers to analyze the dependence structure of the random process. For this purpose, spatial homogeneity and temporal stationarity are widely used assumptions, and many parametric covariance models have been developed under these assumptions. However, these are strong and unrealistic conditions in many cases. In addition, on the sphere, although different statistical approaches from those on Euclidean space should be applied to build a proper covariance model considering its unique characteristics, relevant studies are rare. In this research, we introduce novel parameterized models of the covariance function for spatially non-homogeneous and temporally non-stationary random processes on the sphere. To alleviate the spatial homogeneity assumption and temporal stationarity, and to consider the spherical domain and time domain together, this research will apply the theories of Intrinsic Random Functions (IRF). We also provide a methodology to estimate the associated parameters for the model. Finally, through a simulation study and analysis of a real-world data set about global temperature anomaly, we demonstrate validity of the suggested covariance model with its advantage of interpretability.
\end{abstract}


\section{Introduction} \label{sec:intro}
Although many techniques and theories have been developed to relax the strong assumptions of spatial homogeneity and temporal stationarity for covariance modeling in spatio-temporal data analysis, many of these methods are based on the assumption of Euclidean space, while those for spatio-temporal data on the sphere are relatively sparse. Spatial stochastic processes over the surface of the sphere need different statistical and mathematical approaches relative to those over Euclidean space. Porcu et al.\@ (2016) \cite{porcu2016spatio} asserted that the choices of different distance metrics can seriously influence the estimation quality for the parameter, particularly for a larger region on the sphere. Huang et al.\@ (2011) \cite{huang2011validity} showed how to deal with non-homogeneous spatial data on the sphere by extending the application of intrinsic random functions, a flexible family of non-homogeneous models (Matheron \@ 1973) \cite{Materon1973}.

Another important aspect of parametric covariance modeling is how to parameterize the covariance functions. Obukhov (1947) \cite{Obukhov1947} demonstrated how to parameterize the covariance models on the sphere by assuming of spatial homogeneity. Roy (1969) \cite{Roy1969} developed a covariance model for a spatio-temporal process by adding a stationary time term to Obukhov's model. Porcu et al.\@ (2016) \cite{porcu2016spatio} introduced diverse parametric covariance models for homogeneous and stationary spatio-temporal random processes with comparison of two different types of models: adaptive models from the Euclidean case and models from direct construction on the sphere. The adaptive models are modified from models for Euclidean space. The models by direct construction are only valid for the great circle distance. According to their conclusion, direct construction is a more intuitive and flexible approach whereas the adaptive model requires stronger assumptions.

Our goal is to develop parametric covariance functions through direct construction on the sphere for spatio-temporal random processes. Furthermore, we want to extend this approach to spatially non-homogeneous and temporally non-stationary random processes. In Section \ref{sec:homo_stationary}, we provide the literature review about homogeneous and stationary random processes on the sphere. Section \ref{sec:non_homo_non_stationary} relaxes the assumptions of spatial homogeneity and temporal stationarity through the concepts of intrinsic random functions and intrinsic covariance functions. Finally, Section \ref{sec:data_analysis} demonstrates a methodology of the parameter estimation for the suggested covariance models. This will be illustrated by the simulation study and analysis of a real world data set on global temperature anomaly.

\bigskip

\section{Homogeneous and Stationary Random Processes on the Sphere}\label{sec:homo_stationary}

\subsection{Covariance Functions from Spectral Representation}\label{subsec2.1}

Let $X(P,t)$ be a spatio-temporal process with $P \in \mathbb{S}^2 \text{ and } t \in \mathbb{R}$ where $\mathbb{S}^2$ denotes the surface of the unit sphere.\\
Assume that the process is continuous in quadratic mean. Then, we can expand such a random process through its spectral representation, which is convergent in quadratic mean (Yaglom 1961, Jones 1963, Roy 1969) \cite{Yaglom1961} \cite{jones1963} \cite{Roy1969}:

\begin{align}
X(P,t)= \sum_{\ell=0}^{\infty}\sum_{m=-\ell}^{\ell}Z_{\ell,m}(t)Y_\ell^m(P) \label{eq:spectral_X}
\end{align}
where $$Z_{\ell,m}(t)=\int_{\mathbb{S}^2} X(P,t)Y_\ell^m(P)dP$$
\bigskip

and the $Y_\ell^m(\cdot)$ are the real spherical harmonics such that\\
\begin{equation*}
 \begin{cases}
&Y_\ell^m(P) = \sqrt{\frac{2\ell+1}{2\pi} \frac{(\ell-m)!}{(\ell+m)!}} P_\ell^m(\cos{\phi}) \cos(m\theta)\\
&Y_\ell^0(P) = \sqrt{\frac{2\ell+1}{4\pi}} P_l^0(\cos{\Phi})\\
&Y_\ell^{-m}(P) = \sqrt{\frac{2\ell+1}{2\pi} \frac{(\ell-m)!}{(\ell+m)!}} P_\ell^m(\cos{\phi}) \sin(m\theta)
\end{cases}
\end{equation*}
\bigskip

Each coefficient $Z_{\ell,m}(t)$ is a function of time $t$ and free from location $P$. The term $\theta \in [0,2\pi]$ is longitude, and $\phi \in [-\frac{\pi}{2},\frac{\pi}{2}]$ is latitude; therefore, $P=(\theta, \phi)$. The function $P_\ell^m(\cdot)$ is the associated Legendre polynomials.\\

While the terms `(weakly) homogeneity' and `(weakly) stationarity' are often used interchangeably, in this paper, we reserve `homogeneity' for the spatial domain and `stationarity' for the temporal domain.

\bigskip

\begin{definition}\label{def:homogeneous_stationary}
A spatio-temporal random process $X(\cdot,\cdot)$ is (spatially) homogeneous and (temporally) stationary if the process has

$$E\biggl( X(P,t) \biggl) = E\biggl( X(gP, \tau t) \biggl)$$
and
$$Cov\biggl( X(P,t), X(Q,s) \biggl) = Cov\biggl( X(gP, \tau t), X(gQ, \tau s) \biggl)$$

for any $P,Q \in \mathbb{S}^2$; $t,s \in \mathbb{R}$, $g \in G$, and $\tau \in \mathrm{T}$ where $G$ is a group of rotation operators and $\mathrm{T}$ is a group of shift operators. In other words, the process is spatially rotation invariant and temporally shift invariant.
\end{definition}

\bigskip

Since the coefficients $Z_{\ell,m}(t)$ are random variables, we first need to introduce their appropriate covariance function to obtain the covariance of the process $X(P,t)$. Roy (1969) \cite{Roy1969} showed that if a random process is homogeneous and stationary, then the coefficients $Z_{\ell,m}(t)$ are uncorrelated with each other when the values of $\ell$ and $m$ differ. This independence of the coefficients with respect to the spatial terms allows us to express the covariance structure of a spatio-temporal process in a stationary form, depending only on its time lag given the order of the spherical harmonic $\ell$. As a result, we can express the covariance of the coefficients and the original process as shown by Roy (1969) \cite{Roy1969}\\
 \begin{eqnarray}
a_\ell(h) := Cov\biggl(Z_{{\ell},{m}}(t), Z_{{\ell^{'}},{m^{'}}}(s) \biggl) &=& \mathbbm{1}(\ell,\ell') \mathbbm{1}(m,m') \int_{-\infty}^{\infty} e^{i\omega h} dG_\ell(\omega) \label{eq:al}\\
\phi_0(\overrightarrow{PQ},h) := Cov\biggl(X(P,t), X(Q,s) \biggl) &=& \sum_{\ell=0}^{\infty} \frac{2\ell+1}{4\pi} a_{\ell}(h) P_{\ell}(\cos{\overrightarrow{PQ}}), \label{eq:cov_x}
\end{eqnarray}
where $h$ is $t-s$, $P_\ell(\cdot)$ is a Legendre polynomial, and $G_\ell(\omega)$ is a non-negative measure bounded on $[0,\infty)$.\\

The term $\sum_{\ell=0}^{\infty} \frac{(2\ell+1)}{4\pi} a_\ell(h)$ should be finite because the second moment should be finite. Furthermore, $a_\ell(h)$ should be a positive semi-definite covariance function in terms of $t$. The covariance function of $X(t)$ in \eqref{eq:cov_x} contains $a_\ell(h)$, which is a stationary function with respect to the temporal component. Therefore, we can parameterize the given covariance function by introducing an appropriate form for $a_\ell(h)$. For example, under the specific assumption that $a_\ell(h) = \alpha^\ell g^\ell(h)$, Yarenko (1983) \cite{yadrenko1983spectral} demonstrated that the proper parametric covariance function for the expression in \eqref{eq:cov_x} can be derived using the generating function of Legendre polynomials. Similarly, Porcu et al. (2016) \cite{porcu2016spatio} introduced various parametric covariance models—such as Negative Binomial, Multiquadric, Sine Series, and Sine Power—designed for spatially homogeneous and temporally stationary spatio-temporal random processes on the sphere.\\

\bigskip

\begin{remark}\label{rmk:parametric_models}
Assume $a_\ell(h) = \alpha^\ell g^\ell(h)$ in \eqref{eq:al} where $g(\cdot)$ is a continuous temporal covariance function and $\alpha \in (0,1)$. This assumption leads to the parametric covariance models presented in Table \ref{tab:tab1}. The covariance functions $\phi_0(\psi,h)$ are functions of a great circle distance $\psi$ and time lag $h$, which implies homogeneity and stationarity as desired. \\
\end{remark}

\begin{table}[h!] \label{tab1}
\begin{center}
\caption{\label{tab:tab1} Spatio-temporal parametric covariance models on the sphere. A function $g(\cdot)$ is any covariance function on the real line, and $\gamma$ is a scale parameter. (Yarenko, 1983, \cite{yadrenko1983spectral}, Porcu et al.\@, 2016) \cite{porcu2016spatio} )}
\begin{tabular}{||p{4cm} p{7.5cm} p{4cm}||} 
 \hline
 Model name & Analytic expression & Parameters range\\  [1ex] 
 \hline\hline
 Generating function of Legendre polynomials & $\phi_0(\psi, h) = \frac{\gamma}{4\pi}\frac{ 1 - \alpha^2 g^2(h)}{(1-2 \alpha g(h) \cos{\psi} + \alpha^2 g^2(h))^{3/2}}$ & $\alpha \in (0,1)$\\ [1ex] 
 \hline
 Negative binomial & $\phi_0(\psi, h)= \frac{\gamma}{4\pi}\{\frac{1-\alpha}{1-\alpha g(h) \cos{\psi}}\}^\tau$ & $\alpha \in (0,1)$, $\tau>0$ \\ [1ex] 
 \hline
 Multiquadric &  $\phi_0(\psi, h)= \frac{\gamma}{4\pi}\frac{(1-\alpha)^{2\tau}}{\{1+\alpha^2 - 2 \alpha g(h) \cos{\psi}\}^\tau}$ & $\alpha \in (0,1)$, $\tau>0$ \\ [1ex] 
 \hline
 Sine series & $\phi_0(\psi, h)= \frac{\gamma}{4\pi} e^{\alpha g(h) \cos{d}-1} \{ 1 + \alpha g(h) \cos{\psi} \}/2$ & $\alpha \in (0,1)$ \\  [1ex] 
 \hline
 Sine power & $\phi_0(\psi, h)= \frac{\gamma}{4\pi} \{ 1 - 2^{-\eta} (1 - \alpha g(h) \cos{\psi} ) \}^{\eta/2}$ & $\alpha \in (0,1)$, $\eta \in (0,2]$ \\  [1ex] 
 \hline
 Adapted multiquadric & $\phi_0(\psi, h)= \frac{\gamma}{4\pi}\{\frac{\{1+\alpha g^2(h)\}(1-\alpha)}{1+ \alpha^2 g^2(h) - 2 \alpha g(h) \cos{\psi}}\}^\tau$ & $\alpha \in (0,1)$, $\tau>0$ \\ [1ex] 
 \hline
 Poisson & $\phi_0(\psi, h)= \text{exp}[\lambda\{\cos{(\psi)} \alpha g(h) - 1\}]$ & $\alpha \in (0,1)$, $\lambda>0$ \\ [1ex] 
 \hline
\end{tabular}
\end{center}
\end{table}

\bigskip

\section{Non-Homogeneous and Non-Stationary Spatio-Temporal Random Processes on the Sphere}\label{sec:non_homo_non_stationary}

In this section, we relax the assumptions of spatial homogeneity and temporal stationarity by applying the concept of intrinsic random functions, a flexible family of non-homogeneous models. Intrinsic random functions (Matheron, 1973) \cite{Materon1973}  are generalized random processes, which can be practically used for the representation of spatially non-homogeneous and temporally non-stationary properties. To define an intrinsic random function on the sphere and in time, it is first necessary to understand the allowable measures for each space: the sphere and the real line.

\bigskip

\subsection{Allowable Measures on $\mathbb{S}^2$} \label{subsec:allowable_measures_sphere}

\begin{definition}\label{def:allowable_measure_spatial}
Let $\Lambda^{s}$ denote the dual of $C(\mathbb{S}^2)$, that is, the set of all finite regular signed Borel measures on $\mathbb{S}^2$. For $f \in C(\mathbb{S}^2)$, we define that
\begin{align*}
f(\lambda^s) &= \int_{\mathbb{S}^2} f(x) \lambda^s(dx)\\
f(g \lambda^s) &= \int_{\mathbb{S}^2} f(gx) \lambda^s(dx)
\end{align*}
where $g\lambda^{s}$ is a rotation measure and $g \in G^s$, with $G^s$ being a group of rotation operators.\\
Then, $\lambda^s \in \Lambda^s$ is called an allowable measure of order $\kappa$ on $\mathbb{S}^2$ if it annihilates the spherical harmonics with orders less than $\kappa$ (Huang et al. 2019) \cite{HUANG20197}. That is, we can state that

$$Y_{\ell}^{m}(\lambda^s) = \int_{\mathbb{S}^2} Y_\ell^m(P) \lambda^s(dP) = 0 \quad \text{for} \quad 0 \le \ell < \kappa, \quad |m| \le \ell.$$

We define $\Lambda_\kappa^s$ as the class containing all allowable measures of order $\kappa$. Thus, $\Lambda_{\kappa+1}^s \subset \Lambda_{\kappa}^s \subset \Lambda^s$. 
\end{definition}

\bigskip

\begin{remark} \label{rmk:truncation_measure}
(Huang et al. 2019) \cite{HUANG20197} \\
Define 
$$\lambda_Q^s(dP) := \delta_Q(dP) - \sum_{\ell<\kappa} \sum_{m=-\ell}^{\ell} Y_{\ell}^{m}(Q) Y_\ell^m(dP)$$

where $P,Q \in \mathbb{S}^2$ and $\delta_Q(dP)$ is the Dirac measures on $\mathbb{S}^2$. Then $\lambda_{Q}^s$ is the allowable measure on $\mathbb{S}^2$ that truncates spherical harmonics with orders less than $\kappa$. That is, given $t_0\in \mathbb{R}$, any spatio-temporal random process on the sphere $X(P,t_0)$,
\begin{align*}
    X(\lambda_{Q}^s) &= \int_{\mathbb{S}^2} \sum_{\ell=0}^{\infty} \sum_{m=-\ell}^{\ell} Z_\ell^m(t_0) Y_\ell^m(P) \lambda_{Q}^s(dP)\\
    &= \sum_{\ell=\kappa}^{\infty} \sum_{m=-\ell}^{\ell} Z_\ell^m(t_0) Y_\ell^m(Q)\\
    &=: X_\kappa(Q,t_0)
\end{align*}
\end{remark}

Analogously, we can also explore temporally non-stationary random processes using the properties of the intrinsic random function on $\mathbb{R}$.

\bigskip

\subsection{Allowable Measures on $\mathbb{R}$} \label{subsec:allowable_measures_R}

\begin{definition} (Chiles \@ 1999) \cite{Chiles1999} \label{def:allowable_measure_temporal}
A discrete measure $\lambda^t$ is an allowable measure of order $d$ on $\mathbb{R}$ if it annihilates polynomials of degree less than $d$. That is, a set of $\lambda_i^t$ to $n$ points $x_i$ on $\mathbb{R}$ defines a discrete measure\\
$$\lambda^t = \sum_{i=1}^{n} \lambda_i^t \delta_{x_i}^t \quad \text{where $\delta_{x_i}^t$ is the Dirac measure on $\mathbb{R}$ at $x_i$}$$
and for any $f \in C(\mathbb{R}),$
\begin{align*}
f(\lambda^t) &= \int f(x) \lambda^t(dx)  = \sum_{i=1}^{n} \lambda_i^t f(x_i)\\
f(\tau_h\lambda^t) &= \int f(\tau_h x) \lambda^t(dx)  = \sum_{i=1}^{n} \lambda_i^t f(x_i + h)
\end{align*}
where $\tau_h$ is a shifted measure by any $h \in \mathbb{R}$, and $\tau_h \in G^t$. Then,
\begin{align}
p(\lambda^t) = \int p(x) \lambda^t(dx) = \sum_{i=1}^{n} \lambda_i^t x_i^\ell = 0 \label{eq:allowable_temporal}
\end{align}

where $\ell=0,1,2,\cdots, d-1$ and $p(\cdot)$ is a polynomial function at the degree less than $d$. We denote that $\Lambda_d^t$ is a class of such allowable measures, and then it is clear that $\Lambda_{d+1}^t \subset \Lambda_{d}^t \subset \Lambda^t$.
\end{definition}


\bigskip

\begin{remark} (Kim \& Huang 2024) \cite{kim2025randomprocessesstationaryincrements} \label{rmk:allowable_measure_delta}
    $$\lambda_{\Delta_{h,s}^d}^t = \sum_{k=0}^{d} (-1)^k \binom{d}{k} \delta_{s - kh}, \quad s,h \in \mathbb{R}$$
    
where $s,h \in \mathbb{R}$, and $\delta_{s-kh}$ is the Dirac measure on $\mathbb{R}$. Then, $\lambda_{\Delta_{h,s}^d}$ is a n-th differencing operator. This is an allowable measure of the order $d$ on $\mathbb{R}$. For example, for any spatio-temporal random processes $X(P_0,t)$ given $P_0 \in \mathbb{S}^2$, if d = 1, $X(\lambda_{\Delta_{h,s}}^t) = \Delta_{h} X(P_0, s)  = X(P_0, s) - X(P_0, s-h)$. Likewise, if d=2, $X(\lambda_{\Delta_{h,s}^2}^t) = \Delta_{h}^2 X(P_0, s) = X(P_0, s) - 2 X(P_0, s-h) + X(P_0, s-2h)$.
\end{remark}

\bigskip

Through this allowable measure, we can also define random processes with stationary increments, which is another useful concept for modeling non-stationary random processes on the real line.

\bigskip

\begin{definition} \label{def:random_process_stationary_increment}
(Wei Chen et al. 2024; Yaglom 1958) \cite{Chen2024} \cite{Yaglom1958}\\ 
$\{X(t), t \in \mathbb{R}\}$ is a random process with stationary increments of order $d$ on the real line if $\Delta_{h}^{d}X(t)$ is stationary where \\
$$\Delta_{h}^{d}X(t) = \sum_{k=0}^{d} (-1)^k \binom{d}{k} X(t - kh), \quad t,h \in \mathbb{R}, \quad d \in \mathbb{Z}_{\ge 0}.$$
Additionally, $D^d(\iota; h_1,h_2)$ is its structure function where\\ 

$$D^d(\iota; h_1,h_2) = E\biggl([\Delta_{h_1}^{d}X(t+\iota)] [\Delta_{h_2}^{d}X(t)]\biggl).$$
This is symmetric and positive semi-definite. 
\end{definition}

\bigskip

By using these concepts and the properties of the allowable measures on $\mathbb{S}^2$ and $\mathbb{R}$, an intrinsic random function on $\mathbb{S}^2 \times \mathbb{R}$ can be defined.

\bigskip

\subsection{Intrinsic Random Function for a Spatio-Temporal Random Process on the Sphere}\label{subsec:IRF_spatio_temporal}

\begin{definition}\label{def:irf_spatio_temporal}
A spatio-temporal random process $X(\cdot,\cdot)$ on $\mathbb{S}^2$ is an intrinsic random function of order $(\kappa,d)$, i.e. IRF($\kappa,d$) if
$$E\biggl( X(\lambda^s, \lambda^t) \biggl) = E\biggl( X(g\lambda^s, \tau \lambda^t) \biggl)$$
$$Cov\biggl( X(\lambda_{1}^s, \lambda_{1}^t) ,X(\lambda_{2}^s, \lambda_{2}^t) \biggl) = Cov\biggl( X(g \lambda_{1}^s, \tau \lambda_{1}^t) ,X(g\lambda_{2}^s, \tau \lambda_{2}^t) \biggl)$$

for any allowable measures $\lambda^s, \lambda_{1}^s, \lambda_{2}^s \in \Lambda_{\kappa}^s$ and $\lambda^t, \lambda_{1}^t, \lambda_{2}^t \in \Lambda_{d}^t$. The classes $\Lambda_{\kappa}^s$ and $\Lambda_{d}^t$ are of all possible allowable measures on $\mathbb{S}^2$ and $\mathbb{R}$ with orders $\kappa$ and $d$, respectively; $g$ is any rotation operator and $\tau$ is any shift operator.
\end{definition}

\bigskip

Spatially homogeneous and temporally stationary random process on the sphere is an IRF(0,0). Note that this notation differs slightly from Matheron’s (1973) \cite{Materon1973} in which a stationary process is represented by IRF(-1). In this paper, we use $\kappa \ge 1$ to represent nonstationary random processes, instead of $\kappa \ge 0$.

\bigskip

\begin{definition} \label{def:icf}
    Suppose $X(P,t)$ is an IRF($\kappa$, $d$) on the sphere for any $P\in \mathbb{S}^2$ and any $t \in \mathbb{R}$. Then, the spatially homogeneous and temporally stationary covariance function $Cov\biggl( X(\lambda_1^s,\lambda_1^t),  X(\lambda_2^s,\lambda_2^t) \biggl)$ is called the intrinsic covariance function of order ($\kappa,d$), denoted ICF($\kappa,d$), where $\lambda_{1}^s, \lambda_{2}^s \in \Lambda_{\kappa}^s$ and $\lambda_{1}^t, \lambda_{2}^t \in \Lambda_{d}^t$.
\end{definition}

\bigskip

Now, by applying the concepts of an IRF($\kappa$, $d$) and ICF($\kappa$, d) on the sphere, we can explore the properties of non-homogeneous and non-stationary spatio-temporal random processes. First, we can derive the spectral representation of the process. 

\bigskip

\subsection{Spectral Representation of Non-Homogeneous and Non-Stationary Random Processes on the Sphere}

Let $X(P,t)$ be an IRF($\kappa,d$). In this case, the process is both spatially non-homogeneous and temporally non-stationary. By applying the structure of the spectral representation in \eqref{eq:spectral_X}, we can extend the stochastic process $X(P,t)$ to

\begin{align}
    X(P,t) &= \sum_{\ell=0}^{\infty} \sum_{m=-\ell}^{\ell} Z_{\ell,m}(t) Y_\ell^m(P) \label{eq:spectral_irf(k,d)}
\end{align}

where $P \in \mathbb{S}^2$ and $t \in \mathbb{R}$. Here, the coefficient $Z_{\ell,m}(t)$ is a non-stationary random process on the real line; therefore, it can be expressed by the spectral representation of the random processes of stationary increments with order $d$.

\bigskip

\begin{remark} \label{rmk:coef_spectral}
Let $Z_{\ell,m}(t)$ in \eqref{eq:spectral_irf(k,d)} be a random process with stationary increments of order $d$ and $\ell \ge \kappa$. Then, the following spectral representation can be achieved (Yaglom 1958) \cite{Yaglom1958}.
    \begin{align}
     Z_{\ell,m}(t) &= \int_{-\infty}^{\infty} \biggl(e^{it\omega} -1 - it\omega - \cdots - \frac{(i t \omega)^{n-1}}{(n-1)!} \biggl) \frac{1}{(i\omega)^d} dZ_\ell(\omega) \label{eq:spectral_i(d)}\\
     \Delta_h^d Z_{\ell,m}(t) &= \int_{-\infty}^{\infty} e^{i\omega t} (1-e^{-i h\omega})^d \frac{1}{(i\omega)^d} dZ_\ell(\omega) \label{eq:spectral_differenced} \\
    D(\iota; h_1,h_2) &:= E\biggl([\Delta_{h_1}^{d}Z_{\ell,m}(t+\iota)] [\Delta_{h_2}^{d}Z_{\ell',m'}(t)]\biggl) \notag\\
    &=\int_{-\infty}^{\infty} e^{i\omega \iota} (1-e^{-i h_1 \omega})^d (1-e^{i h_2 \omega})^d \frac{1}{\omega^{2d}} dF_{\ell}(\omega) \mathbbm{1}(\ell,\ell') \mathbbm{1}(m, m')  \label{eq:spectral_structure_id}
\end{align}

where $F_\ell(\omega)$ is a bounded non-decreasing function and $E(|dZ_\ell(\omega)|^2) = dF_\ell(\omega)$ for $Z_\ell(\omega)$, a random function with uncorrelated increments. For an IRF($\kappa,d$), $E\biggl([Z_{\ell,m}(t)] [Z_{\ell',m'}(s)]\biggl) = 0$ if $\ell \ne \ell'$ and $m \ne m'$ for $\ell', \ell' \ge \kappa$.

\end{remark}

\bigskip

\begin{remark} (Kim \& Huang 2024) \cite{kim2025randomprocessesstationaryincrements} \label{rmk:irf_rpsi} 
Assume that the spectral density of $Z_{\ell,m}(t)$ in \eqref{eq:spectral_i(d)} exists; therefore, $dF_\ell(\omega)$ in \eqref{eq:spectral_structure_id} can be expressed as $f_\ell(\omega)d\omega$. Then, for $d \ge 1$ and $d \in \mathbb{Z}$, a stochastic process $Z_{\ell,m}(t)$ is an intrinsic random function of order $d$ on the real line if and only if $Z_{\ell,m}(t)$ is a random process with stationary increments order $d$.
\end{remark}

\bigskip

\begin{theorem} \label{thm:icf_spatio_temporal}
A spatio-temporal random process on the sphere $X(P,t)$ is an intrinsic random function with orders ($\kappa,d$) if and only if its truncated and differenced process 
$$\Delta_h^d X_\kappa(P,t) = \sum_{\ell=\kappa}^{\infty} \sum_{m=-\ell}^{\ell} \Delta_{h}^{d} Z_{\ell,m}(t) Y_\ell^m(P), \quad P \in \mathbb{S}^2, \quad t,h \in \mathbb{R}$$ 
is spatially homogeneous and temporally stationary.
\end{theorem}

\bigskip

\begin{proof}
($\Rightarrow$) Suppose that a spatio-temporal random process $X(\cdot,\cdot)$ on $\mathbb{S}^2$ is an IRF($\kappa,d$). Denote allowable measures 
\begin{align*}
&\lambda_{Q}^s(dP) = \delta_Q(dP) - \sum_{\ell<\kappa} \sum_{m=-\ell}^{\ell} Y_{\ell}^{m}(Q) Y_\ell^m(dP)\\
&\lambda_{\Delta_{h,s}^d}^t = \sum_{k=0}^{d} (-1)^k \binom{d}{k} \delta_{s - kh}, \quad s,h \in \mathbb{R}\\
\text{where } &P, Q \in \mathbb{S}^2; \quad s,h \in \mathbb{R}; \quad \delta_Q(dP) \quad  \text{and} \quad \delta_{s-kh} \text{ are the Dirac measures on $\mathbb{S}^2$ and $\mathbb{R}$, respectively.}\\
\end{align*}

By Remark \ref{rmk:truncation_measure} and Remark \ref{rmk:allowable_measure_delta}, $\lambda_{Q}^s \in \Lambda_{\kappa}^s$ and $\lambda_{\Delta_{h,s}^d}^t \in \Lambda_{d}^t$. Therefore,
\begin{align*}
X(\lambda_{Q}^s,\lambda_{t_s}) &= \int_{\mathbb{R}} \int_{\mathbb{S}^2} X(P,t) \lambda_{Q}^s(dP) \lambda_{t_s}(dt)\\
&=\int_{\mathbb{R}} \int_{\mathbb{S}^2} \sum_{\ell=0}^{\infty} \sum_{m=-\ell}^{\ell} Z_{\ell,m}(t) Y_\ell^m(P) \lambda_{Q}^s(dP) \lambda_{t_s}(dt)\\
&\text{Since $\lambda_s$ and $\lambda_t$ are finite measures and $Z_{\ell,m}(t)$ is integrable,}\\
&\text{by Fubini's theorem,}\\
&= \sum_{\ell=0}^{\infty} \sum_{m=-\ell}^{\ell} \int_{\mathbb{R}} Z_{\ell,m}(t) \lambda_{t_s}(dt) \int_{\mathbb{S}^2} Y_\ell^m(P) \lambda_{Q}^s(dP)\\
&\text{By the definitions of $\lambda_{Q}^s$ and Remark \ref{rmk:truncation_measure},}\\
&= \sum_{\ell=\kappa}^{\infty} \sum_{m=-\ell}^{\ell} Y_\ell^m(Q) \int_{\mathbb{R}} Z_{\ell,m}(t) \lambda_{t_s}(dt) \\
&\text{By the definitions of $\lambda_{t_s}$ and Remark \ref{rmk:allowable_measure_delta},}\\
&= \sum_{\ell=\kappa}^{\infty} \sum_{m=-\ell}^{\ell} Y_\ell^m(Q) \Delta_{h}^{d} Z_{\ell,m}(s)\\
&= \Delta_{h}^{d} X_\kappa(Q,s)\\
\end{align*}

Likewise,\\
$$X(g \lambda_{Q}^s, \tau \lambda_{t_s}) = \Delta_{h}^{d} X_\kappa(g Q, \tau s)$$

By the definition of IRF($\kappa$,$d$),\\
$$E\biggl( X(\lambda_{Q}^s, \lambda_{t_s}) \biggl) = E\biggl( X(g\lambda_{Q}^s, \tau \lambda_{t_s}) \biggl)$$
$$Cov\biggl( X(\lambda_{s_P}, \lambda_{t_t}) ,X(\lambda_{Q}^s, \lambda_{t_s}) \biggl) = Cov\biggl( X(g \lambda_{s_P}, g \lambda_{t_t}) ,X(g\lambda_{Q}^s, \tau \lambda_{t_s}) \biggl)$$

Therefore,
$$E\biggl( \Delta_{h}^{d} X_\kappa(Q, s) \biggl) = E\biggl( \Delta_{h}^{d} X_\kappa(g Q, \tau s) \biggl)$$
$$Cov\biggl( \Delta_{h}^{d} X_\kappa(P, t), \Delta_{h}^{d} X_\kappa(Q, s) \biggl) = Cov\biggl( \Delta_{h}^{d} X_\kappa(g P, \tau t), \Delta_{h}^{d} X_\kappa(g Q, \tau s) \biggl)$$

That is, $\Delta_{h}^{d} X_\kappa(Q, s)$ is spatially homogeneous and temporally stationary.

\bigskip

($\Leftarrow$) Suppose $X(P,t)$ is a spatio-temporal random process on the sphere, and its truncated and differenced process $\Delta_h^d X_\kappa(P,t)$ is homogeneous and stationary. We can have spectral representations such that
$$\Delta_h^d X_\kappa(P,t) = \sum_{\ell=\kappa}^{\infty} \sum_{m=-\ell}^{\ell} \Delta_h^d Z_{\ell,m}(t) Y_\ell^m(P)$$
$$X(P,t) = \sum_{\ell=0}^{\infty} \sum_{m=-\ell}^{\ell} Z_{\ell,m}(t) Y_\ell^m(P)$$
Then, for any $\lambda^s \in \Lambda_{\kappa}^s$ and $\lambda^t \in \Lambda_d^t$,

\begin{align*}
    &E \biggl( X(\lambda^s, \lambda^t) \biggl) = E \int_\mathbb{R} \int_{\mathbb{S}^2} X(P,t) \lambda^s(dP) \lambda^t(dt)\\
    &=E \int_\mathbb{R} \int_{\mathbb{S}^2} \biggl\{ \sum_{\ell=0}^{\kappa-1} \sum_{m=-\ell}^{\ell} Z_{\ell,m}(t) Y_\ell^m(P) \biggl\} \lambda^s(dP) \lambda^t(dt) + E \int_\mathbb{R} \int_{\mathbb{S}^2} X_\kappa(P,t) \lambda^s(dP) \lambda^t(dt)\\
    &\text{Since $\lambda^s \in \Lambda_{\kappa}^s$,}\\
    &=E \int_\mathbb{R} \int_{\mathbb{S}^2} X_\kappa(P,t) \lambda^s(dP) \lambda^t(dt)\\
    &\text{Since $\lambda^s$ and $\lambda^t$ are finite measures and $X(P,t)$ is continuous in quadratic mean, by Fubini's theorem,}\\
    &= \int_\mathbb{R} \int_{\mathbb{S}^2} E\biggl( X_\kappa(P,t) \biggl) \lambda^s(dP) \lambda^t(dt)\\
    &\text{By the homogeneity of $X_\kappa(P,t)$,}\\
    &= \int_\mathbb{R} \int_{\mathbb{S}^2} E\biggl( X_\kappa(gP,t) \biggl) \lambda^s(dP) \lambda^t(dt)\\
    &\text{By Fubini's theorem again,}\\
    &= \int_{\mathbb{S}^2} E \biggl( \int_{\mathbb{R}} X_\kappa(gP,t) \lambda^t(dt) \biggl) \lambda^s(dP) = \int_{\mathbb{S}^2} E \biggl( \sum_i \lambda_i X_\kappa(gP,t_i) \biggl) \lambda^s(dP)\\
    &= \int_{\mathbb{S}^2} E \biggl( \sum_i \lambda_i \sum_{\ell=\kappa}^{\infty} \sum_{m=-\ell}^{\ell} Z_{\ell,m}(t_i) Y_\ell^m(gP) \biggl) \lambda^s(dP) = \int_{\mathbb{S}^2} \sum_{\ell=\kappa}^{\infty} \sum_{m=-\ell}^{\ell} Y_\ell^m(gP) E \biggl( \sum_i \lambda_i Z_{\ell,m}(t_i) \biggl) \lambda^s(dP)\\
    &= \int_{\mathbb{S}^2} \sum_{\ell=\kappa}^{\infty} \sum_{m=-\ell}^{\ell} Y_\ell^m(gP) E \biggl( Z_{\ell,m}(\lambda^t) \biggl) \lambda^s(dP)\\
    &\text{By Remark \ref{rmk:irf_rpsi}, given $\ell$ and $m$, a random process with stationary increments, $Z_{\ell,m}(t)$ is an IRF($d$) with respect to $t\in \mathbb{R}$.}\\
    &\text{Therefore, $E\biggl( Z_{\ell,m}(\lambda^t) \biggl)$ = $E \biggl(Z_{\ell,m}(\tau \lambda^t) \biggl)$ for any allowable $\lambda^t \in \Lambda_{d}^t$ and any shift operator $\tau$. Hence,}\\ 
    &= \int_{\mathbb{S}^2} \sum_{\ell=\kappa}^{\infty} \sum_{m=-\ell}^{\ell} Y_\ell^m(gP) E \biggl( Z_{\ell,m}(\tau \lambda^t) \biggl) \lambda^s(dP) = \int_{\mathbb{S}^2} \sum_{\ell=\kappa}^{\infty} \sum_{m=-\ell}^{\ell} Y_\ell^m(gP) E \biggl( \sum_i \lambda_i Z_{\ell,m}(\tau t_i) \biggl) \lambda^s(dP)\\
    &= \int_{\mathbb{S}^2} \sum_{\ell=\kappa}^{\infty} \sum_{m=-\ell}^{\ell} Y_\ell^m(gP) E \biggl( \int_{\mathbb{R}} Z_{\ell,m}(\tau t) \lambda^t(dt) \biggl) \lambda^s(dP)\\
    &\text{By Fubini's theorem,}\\
    &= E \biggl(  \int_{\mathbb{R}} \int_{\mathbb{S}^2} \sum_{\ell=\kappa}^{\infty} \sum_{m=-\ell}^{\ell} Y_\ell^m(gP) Z_{\ell,m}(\tau t) \lambda^s(dP) \lambda^t(dt) \biggl) = E \biggl(  \int_{\mathbb{R}} \int_{\mathbb{S}^2} \sum_{\ell=0}^{\infty} \sum_{m=-\ell}^{\ell} Y_\ell^m(gP) Z_{\ell,m}(\tau t) \lambda^s(dP) \lambda^t(dt) \biggl)\\
    &= E \biggl( \int_{\mathbb{R}} \int_{\mathbb{S}^2} X(g P, \tau t)  \lambda^s(dP) \lambda^t(dt) \biggl) = E \biggl(X(g \lambda^s, \tau \lambda^t) \biggl)
\end{align*}

As a result, $E \biggl( X(\lambda^s, \lambda^t) \biggl) = E \biggl(X(g \lambda^s, \tau \lambda^t) \biggl)$.
In the same manner, we can show that
$$Cov \biggl( X(\lambda_1^s, \lambda_1^t) \biggl) = Cov \biggl(X(g \lambda_2^s, \tau \lambda_2^t) \biggl)$$
for any $\lambda_1^s, \lambda_2^s \in \Lambda_{\kappa}^s$ and $\lambda_1^t, \lambda_2^t \in \Lambda_d^t$. 
Thus, $X(P,t)$ is an IRF($\kappa$,d).
\end{proof}

\bigskip

Now, suppose that $X(\cdot,\cdot)$ is an IRF($\kappa,d$). Then, for the allowable measures $\lambda_{Q}^s \in \Lambda_{\kappa}^s$ and $\lambda_{\Delta_{1,s}^d}^t \in \Lambda_{d}^t$, its intrinsic covariance function (ICF($\kappa,d$)) is:

\begin{align*}
&\phi_{\kappa,d} \biggl(\overrightarrow{PQ}, h \biggl) := Cov \biggl( X(\lambda_{P}^s, \lambda_{\Delta_{1,t}^d}^t), \quad X(\lambda_{Q}^s, \lambda_{\Delta_{1,s}^d}^t) \biggl) \\
&= Cov \biggl( \Delta_1^d X_\kappa(P,t), \quad \Delta_1^d X_\kappa(Q,s) \biggl) \\
&= \sum_{\ell=\kappa}^\infty \frac{2\ell+1}{4\pi} a_{\ell}(h) P_{\ell}(\cos{\overrightarrow{PQ}}) \\
&= \sum_{\ell=0}^\infty \frac{2\ell+1}{4\pi} a_{\ell}(h) P_{\ell}(\cos{\overrightarrow{PQ}}) -  \sum_{\ell=0}^{\kappa-1} \frac{2\ell+1}{4\pi} a_{\ell}(h) P_{\ell}(\cos{\overrightarrow{PQ}}) \\
&= \phi_{0}\biggl(\overrightarrow{PQ}, h \biggl) -  \sum_{\ell=0}^{\kappa-1} \frac{2\ell+1}{4\pi} a_{\ell}(h) P_{\ell}(\cos{\overrightarrow{PQ}})
\end{align*}
By adding a scale parameter $\gamma_0$, this yields
\begin{align}
&\phi_{\kappa,d} \biggl(\overrightarrow{PQ}, h \biggl) = \gamma_0 \biggl[ \phi_{0}\biggl(\overrightarrow{PQ}, h \biggl) -  \sum_{\ell=0}^{\kappa-1} \frac{2\ell+1}{4\pi} a_{\ell}(h) P_{\ell}(\cos{\overrightarrow{PQ}}) \biggl] \label{eq:icsf}
\end{align}
where $a_\ell(h)$ is a stationary time series covariance function and $\phi_0(\cdot,\cdot)$ is a homogeneous and stationary spatio-temporal covariance function. Therefore, covariance functions presented in Table \ref{tab1} can be used for the structure of $\phi_0(\cdot,\cdot)$.

\bigskip

\subsection{Covariance Functions of an IRF($\kappa, d$) on the Sphere} \label{subsec:cov_irf_kappa_d}

Huang et al. (2019) \cite{HUANG20197} proposed that a covariance function of an intrinsic random function on the sphere can be derived by applying a reproducing kernel Hilbert space suggested by Leversely (1999) \cite{Levesley1999}. To introduce an appropriate covariance model for an IRF($\kappa, d$) on the sphere, we can also utilize the structure of this spatially non-homogeneous covariance function, along with the covariance function of non-stationary random processes, as discussed in Remark \ref{rmk:coef_spectral}. Therefore, by introducing the new basis of the nil space, $q_1(\cdot), q_2(\cdot), \dots, q_{\kappa^2}(\cdot) \in N = \text{span}\{ Y_\ell^m(P): \quad 0 \le \ell < \kappa, \quad -\ell \le m \le \ell \}$ and a set of distinct points $\{\tau_1, \tau_2, \dots, \tau_{\kappa^2}\} \in \mathbb{S}^2$  such that $q_\nu(\tau_\mu) = I(\nu, \mu)$ for $1 \le \nu, \mu \le \kappa^2$, we can get the following non-homogeneous and non-stationary covariance function of $IRF(\kappa,d)$ on the sphere:\\

\begin{align}
&R(P,Q,t,s) := Cov\biggl(X(P,t), X(Q,s)\biggl) \notag\\
&= \gamma_0^2 \phi_{\kappa} \biggl(\overrightarrow{PQ}, (t,s) \biggl) + \sum_{\nu=1}^{\kappa^2} \sum_{\mu=1}^{\kappa^2} \gamma_\nu \gamma_\mu \phi_{\kappa} \biggl(\overrightarrow{\tau_{\nu}\tau_{\mu}}, (t,s) \biggl) q_{\nu}(P) q_{\mu}(Q) + \sum_{\nu=1}^{\kappa^2} {\gamma^2_\nu}  q_{\nu}(P) q_{\nu}(Q) \notag\\
&- \sum_{\nu=1}^{\kappa^2} \gamma_0 \gamma_\nu \phi_{\kappa} \biggl(\overrightarrow{Q\tau_{\nu}}, (t,s) \biggl) q_{\nu}(P) - \sum_{\nu=1}^{\kappa^2} \gamma_0 \gamma_\nu \phi_{\kappa} \biggl(\overrightarrow{P\tau_{\nu}}, (t,s) \biggl) q_{\nu}(Q)  \label{eq:non-homogeneous_non-stationary_cov}\\
&\notag\\
&\text{where } \gamma_0, \gamma_\nu, \gamma_\mu \text{ are scale parameters, and} \notag\\
&\phi_{\kappa} \biggl(\overrightarrow{PQ}, (t,s) \biggl) = Cov\biggl(X_\kappa(P,t), X_\kappa(Q,s) \biggl) \notag\\ 
&= \sum_{\ell=\kappa}^{\infty} \frac{2\ell+1}{4\pi} P_\ell(\cos{\overrightarrow{PQ}}) \int_{-\infty}^{\infty} \biggl(e^{it\omega} -1 - it\omega - \cdots - \frac{(i t \omega)^{n-1}}{(n-1)!} \biggl) \biggl(e^{-is\omega} -1 + is\omega - \cdots - \frac{(-i s \omega)^{n-1}}{(n-1)!} \biggl) \frac{1}{\omega^{2d}} dF_\ell(\omega)\label{eq:phi_kappa}
\end{align}
where $P,Q \in \mathbb{S}^2, \quad t,s \in \mathbb{R}$, and $dF_\ell(\omega)$ is a non-decreasing and bounded measure presented in Remark \ref{rmk:coef_spectral}. 

By the Kolmogorov existence theorem, there exists a Gaussian random process having \eqref{eq:non-homogeneous_non-stationary_cov} as its covariance function; hence, we can derive a covariance function of an intrinsic random function of its order $(\kappa,d)$.

\bigskip

\begin{theorem} \label{thm:positive_definite}
The function $R(P,Q,t,s)$ in \eqref{eq:non-homogeneous_non-stationary_cov} is positive semidefinite.
\end{theorem}
\begin{proof}
\begin{align*}
&\text{We can rewrite  \eqref{eq:non-homogeneous_non-stationary_cov} as :}\\
&Cov\biggl(X(P,t), X(Q,s)\biggl)\\
=& \sum_{\ell=\kappa}^{\infty} \sum_{m=-\ell}^{\ell} \int_{-\infty}^{\infty} \biggl(e^{it\omega} -1 - it\omega - \cdots - \frac{(i t \omega)^{n-1}}{(n-1)!} \biggl) \biggl(e^{-is\omega} -1 + is\omega - \cdots - \frac{(-i s \omega)^{n-1}}{(n-1)!} \biggl) \frac{1}{\omega^{2d}} dF_\ell(\omega)\\
&\biggl[ \gamma_0^2 Y_\ell^m(P) Y_\ell^m(Q) - \sum_{\nu=1}^{\kappa^2}  \gamma_0 \gamma_\nu Y_\ell^m(P) Y_\ell^m(\tau_\nu) q_\nu(Q) - \sum_{\nu=1}^{\kappa^2} \gamma_0 \gamma_\nu Y_\ell^m(\tau_nu) Y_\ell^m(Q) q_\nu(P)\\ 
&+ \sum_{\nu=1}^{\kappa^2} \sum_{\mu=1}^{\kappa^2} \gamma_\nu \gamma_\mu Y_\ell^m(\tau_\nu) Y_\ell^m(\tau_\mu) q_\nu(P) q_\mu(Q) \biggl] +  \sum_{\nu=1}^{\kappa^2}\gamma_\nu^2 q_{\nu}(P) q_{\nu}(Q) \\
\\
&\text{Therefore, for any $n \in \mathbb{N}$,  $c_{i_1},c_{i_2} \in \mathbb{C}$,  $x_{i_1}, x_{i_2} \in \mathbb{S}^2$ and $t_{i_1},t_{i_2} \in \mathbb{R}$ where $i_1,i_2=1,2,3 \dots$,}\\
&\sum_{i_1=1}^{n} \sum_{i_2=1}^{n} c_{i_1} \bar{c}_{i_2} Cov\biggl(X(x_{i_1},t_{i_1}), X(x_{i_2},t_{i_2})\biggl) \\
=& \sum_{\ell=\kappa}^{\infty} \sum_{m=-\ell}^{\ell} \int_{-\infty}^{\infty} dF_{\ell}(\omega) \Biggl| \sum_{i_1=1}^n c_{i_1} \biggl(e^{i t_{i_1} \omega} -1 - i t_{i_1} \omega - \cdots - \frac{(i t_{i_1} \omega)^{n-1}}{(n-1)!} \biggl) \frac{1}{(i \omega)^{d}}  \biggl\{ \gamma_0 Y_\ell^m(x_{i_1}) - \sum_{\nu=1}^{\kappa^2} \gamma_{\nu} Y_\ell^m(\tau_\nu) q_\nu(x_{i_1}) \biggl\} \Biggl|^2\\
&+ \sum_{\nu=1}^{\kappa^2} \Biggl| \gamma_{\nu} \sum_{i_1=1}^n c_{i_1} q_\nu(x_{i_1}) \Biggl|^2  \ge 0\\
\\
&\text{where $F_\ell(\omega)$ is a non-negative measure.}\\
\end{align*}
\end{proof}

\bigskip

\section{Data Analysis} \label{sec:data_analysis}

In this section, a methodology of the parameter estimation is demonstrated for the suggested intrinsic covariance functions through numerical optimization. This will be illustrated by the simulation study and analysis of a real world data set about global temperature anomaly, which was collected by The National Space Science and Technology Center (NSSTC, located at the University of Alabama, Huntsville). 

\subsection{Simulation Study} \label{subsec:simulation}

\subsubsection{Simulation Design} \label{subsubsec:simuldesign}

For the simulation, IRF(1,1) with the model of generating function of Legendre polynomials in Table \ref{tab1} is tried. The procedure for simulating the parameter estimation of the intrinsic covariance function of an IRF(1,1) is as follows:
\bigskip
  \begin{enumerate}
  \item Generate $X(\cdot, \cdot)$, an IRF(1,1), from multivariate normal MVN(0, R($P,Q,t,s$)). The non-homogeneous and non-stationary covariance function R($P,Q,t,s$) in \eqref{eq:non-homogeneous_non-stationary_cov} is
  
  \begin{align}
  R(P,Q,t,s) &= \gamma_0^2 \phi_{1} \biggl(\overrightarrow{PQ}, (t,s) \biggl) + \gamma_1^2 \phi_{1} \biggl(0, (t,s) \biggl) + {\gamma^2_1} \notag\\
  &- \gamma_0 \gamma_1 \phi_{1} \biggl(\overrightarrow{Q\tau}, (t,s) \biggl) - \gamma_0 \gamma_1 \phi_{1} \biggl(\overrightarrow{P\tau}, (t,s) \biggl) \label{eq:eq:covariance_irf(1,1)}
  \end{align}
 \begin{align*}
&\text{where } \gamma_0 \text{ and } \gamma_1 \text{ are scale parameters},\\
&\phi_{1} \biggl(\overrightarrow{PQ}, (t,s) \biggl) = Cov\biggl(X_1(P,t), X_1(Q,s) \biggl)\\ 
&= \sum_{\ell=1}^{\infty} \frac{2\ell+1}{4\pi} P_\ell(\cos{\overrightarrow{PQ}}) \int_{-\infty}^{\infty} (e^{i\omega t} - 1)(e^{-i\omega s} - 1)dG_\ell(\omega) - \frac{1}{4\pi}, \quad P,Q \in \mathbb{S}^2, \quad t,s \in \mathbb{R} \\
&\text{where } dG_\ell(\omega) = \frac{1}{(i\omega)^{2d}}dF_\ell(\omega) \text{, which is a non-decreasing and bounded measure.}
\end{align*}
  
  \item Compute the homogeneous and stationary random process $\Delta_1 X_1(\cdot, \cdot)$ by truncating and differencing from the $X(P,t)$ generated from the step 1.\\
  
  \item Estimate the parameters of an intrinsic covariance structure function in \eqref{eq:icsf} through Method of Moments (MoM) estimates and a numerical optimization with respect to the sum of squares between their theoretical values and MoM estimates. For the numerical optimization, the quasi-Newton method is used through the R function \textit{nlminb} (Gay 1990) \cite{nlminb}. For this simulation setting, the model of generating functions of Legendre polynomials in Table \ref{tab1} is tried. For the stationary covaraince function $a_\ell(h)$ in Remark \ref{rmk:parametric_models}, $a_\ell(h)=\alpha^\ell e^{-\ell \beta |h|}$ is used. Thus, the intrinsic covariance function is given as
  \begin{align}
  \phi_{1,1} \biggl(\overrightarrow{PQ}, h \biggl) &= \gamma_0 \biggl[ \phi_{0}\biggl(\overrightarrow{PQ}, h \biggl) - \frac{1}{4\pi} \biggl] \label{eq:icf_11}
  \end{align}
    where $\phi_{0} \biggl(\psi, h \biggl) = \frac{1}{4\pi}\frac{\gamma (1 - \alpha^2 g^2(h))}{(1-2 \alpha g(h) \cos{\psi} + \alpha^2 g^2(h))^{3/2}}$ and $\alpha g(h) =  \alpha e^{-\beta |h|}$.
\end{enumerate}

\bigskip

To get the MoM estimates, let $\Delta_1 X_1 (P,t)$ be a spatio-temporal process and assume $\mu=0$. Then, its MoM estimator is (Bussberg 2020) \cite{bussberg2020environmental}:\\
\begin{eqnarray}
\hat{\phi}_{1,1}(\psi_i, h_j) = \frac{1}{\lvert N_{\psi_i, h_j} \lvert}\sum_{\{(P,t),(Q,s)\} \in N_{\psi_i, h_j}} \Delta_1 X_1(P,t) \Delta_1 X_1(Q,s) \label{eq:mom}
\end{eqnarray} 

where  $N_{\psi_i, h_j} = \biggl\{ \{(P,t),(Q,s)\} \quad | \quad \psi_i - \epsilon \le \psi(P,Q) \le \psi_i+\epsilon, \quad \lvert t-s \lvert = \lvert h_j \lvert \biggl\}$. Here, $\psi(P,Q)$ denotes the great circle distance between $P$ and $Q$ for $P,Q \in \mathbb{S}^2$.

\bigskip

We want to find parameter estimates that minimize the sum of squares:
\begin{align*}
    L(p, \gamma_0) =  \sum_{i=1}^{k_1} \sum_{j=1}^{k_2} \left[\hat{\phi}_{1,1}(\psi_i, h_j) - \phi_{1,1}(\psi_i, h_j, p, \gamma_0)\right]^2.
\end{align*}
where $p$ is the parameters of $a_\ell(\cdot)$ and $\gamma_0$ is the scale parameter in \eqref{eq:icsf}.

That is, we want $\hat{p}$ and $\hat{\gamma_0}$ that minimize the above loss function. The structure of the intrinsic covariance function $\phi_{1,1}(\psi_i, h_j, p, \gamma_0)$ is given in \eqref{eq:icf_11}. To parameterize covariance models, we need to introduce proper intrinsic covariance functions $\phi_{\kappa,d}(\cdot,\cdot)$ and a temporal covariance function $a_\ell(\cdot)$ satisfying the condition of the form $\alpha^\ell g^\ell(\cdot)$ in Remark \ref{rmk:parametric_models}. For instance, if we use $a_\ell(h) = \alpha^\ell e^{-\beta \ell |h|}$, which exponentially decays with $h$, then the parameter of the intrinsic covariance function we have to estimate would be $p=(\alpha, \beta)$. 

\bigskip

\subsubsection{Preliminary Study} \label{subsubsec:preliminary}

Before creating a simulation dataset, we conduct a preliminary study to understand the role of each parameter $\alpha$ and $\beta$ for $a_\ell(h)=\alpha^\ell e^{-\beta \ell |h|}$ in the covariance function \eqref{eq:icf_11}. In Figure \ref{fig:preliminary_parameters}, We can see correlations are decreasing as both great circle distances and temporal lags increase as desired. Figure \ref{fig:preliminary_parameters} also explains how $\alpha$ and $\beta$ control decay rates of the covariance. The upper two plots compare covariance functions having different parameter $\alpha$ values when $\beta$ is fixed as 0.1. The left one shows how covariance values change by time lag with a fixed great circle distance $\psi=0$; likewise, the right one does for different great circle distances with a fixed value of time lag $h=0$. In the same sense, the two plots on the bottom provide a comparison of the covariance values among different parameter values of $\beta$ with fixed $\alpha=0.8$. From this visualization in Figure \ref{fig:preliminary_parameters}, we can conclude that as the values of the parameters $\alpha$ and $\beta$ increase, the covariance values decay more rapidly. This fact illustrates a very important property that we do not need to consider every observation for the computation of our estimate in \eqref{eq:mom}. In other words, when compute the estimate, we only need to consider pairs of a dataset whose spatial or temporal distances are close enough so that their correlations are significantly differently from zero. For example, as shown in Figure \ref{fig:preliminary_parameters}, when $\alpha=0.8$ and $\beta=0.3$ (bottom green), for pairs whose spherical distances are bigger than 1.5 and time lags are bigger than 4, their correlations are almost 0. Hence, when we compute the estimate, these pairs can be ignored. On the other hand, when $\alpha=0.3$ and $\beta=0.1$, correlations decay so slow that we cannot see they become zero even for pairs whose time lags and distances are great. Thus, we aim to use reasonable values for $\alpha$ and $\beta$ that cause correlations to decay at an appropriate rate—not too slowly, nor too quickly—for this simulation. The value of the scale parameter $\gamma$ is negligible for the simulation data, so we fixed $\gamma=1$.\\

\begin{figure}[h!]
\centering
  \includegraphics [width=12cm, height=10cm] {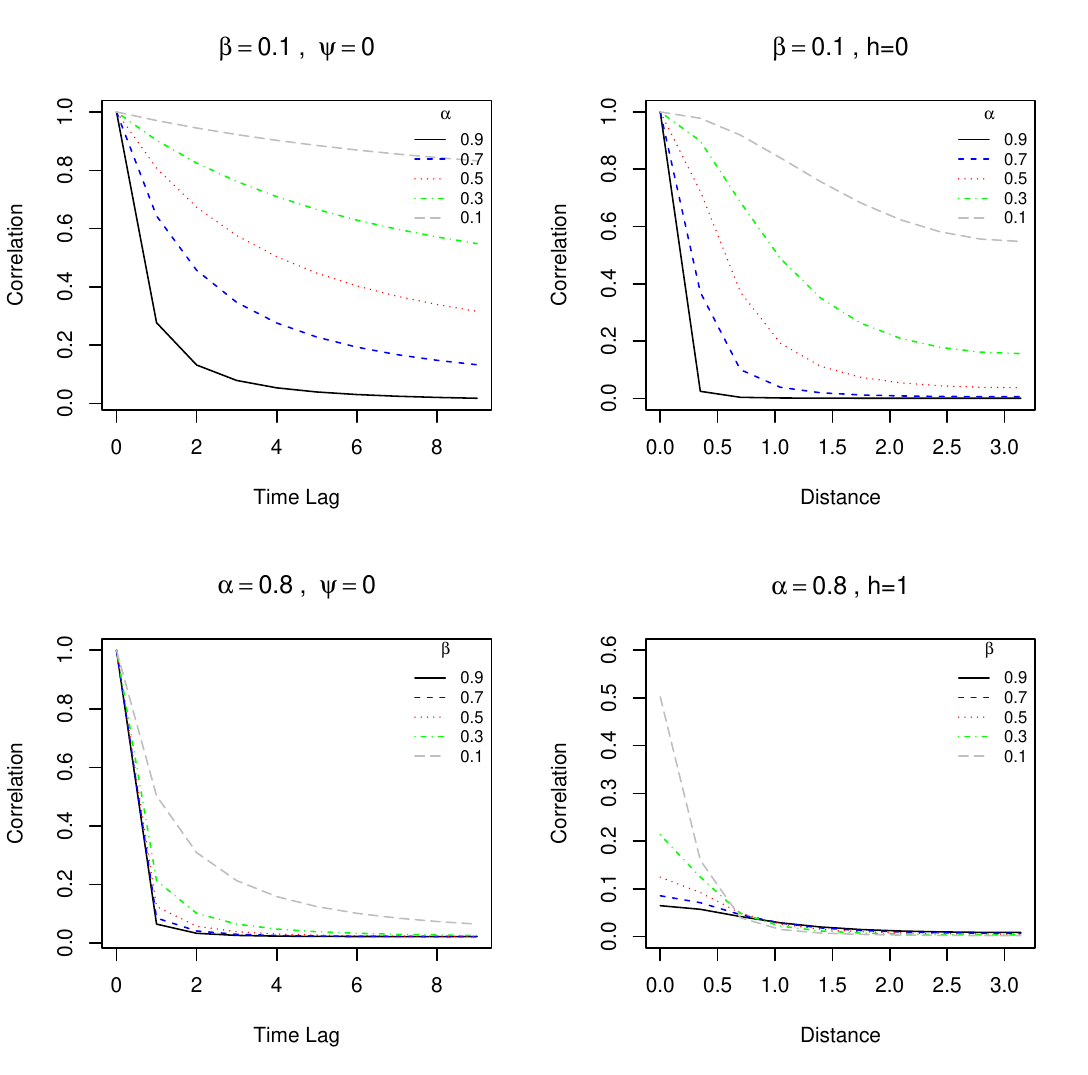}
  \caption{Correlations based on different parameter values. The model of generating functions of Legendre polynomials in Table \ref{tab:tab1} is used with $a_\ell(h) = \alpha^\ell g^\ell(h) =\alpha^\ell e^{-\ell \beta |h|}$.}
  \label{fig:preliminary_parameters}
\end{figure}

\bigskip

\subsubsection{Simulation Result}\label{subsubsec:sim_result}

In this simulation setting, each simulated dataset consists of 1,000 locations and 20 temporal points. In addition, we assign the scale parameter $\gamma_1=\frac{1}{2\sqrt{\pi}}$ in \eqref{eq:eq:covariance_irf(1,1)}. This is because the basis function of the nil space $q_\nu(\cdot)$ tends to be much larger than the values of spherical harmonics $Y_\ell^m(\cdot)$. For example, when $\kappa=1$, $q_1(\cdot)=1$, but $Y_0^0(\cdot)=\frac{1}{2\sqrt{\pi}} \approx 0.2821$. Thus, without adjusting this scale parameter, the non-homogeneous components of the random process would be overweighted. Figure \ref{fig:splot} and Figure \ref{fig:tplot} compare the MoM estimates, fitted values, and true values of the covariance, given that either time lags or great circle distances are fixed with $\alpha=0.80$, $\beta=0.10$, and $\gamma_0=1$. The "true value" refers to the actual values of the intrinsic covariance function computed with the true parameter values. The "fitted values" are obtained by plugging the parameter estimates into the intrinsic covariance function formula in \eqref{eq:icf_11}, while the MoM values come from \eqref{eq:mom}. From this comparison, we can conclude that our covariance estimates perform well, as the fitted values of the intrinsic covariance function are very close to the true values, regardless of the distances or lags. Furthermore, as expected, we observe that the values of the intrinsic covariance function (ICF) decay as time lags and great circle distances increase.

\begin{figure}[h!]
\centering
  \includegraphics [width=13cm, height=11cm] {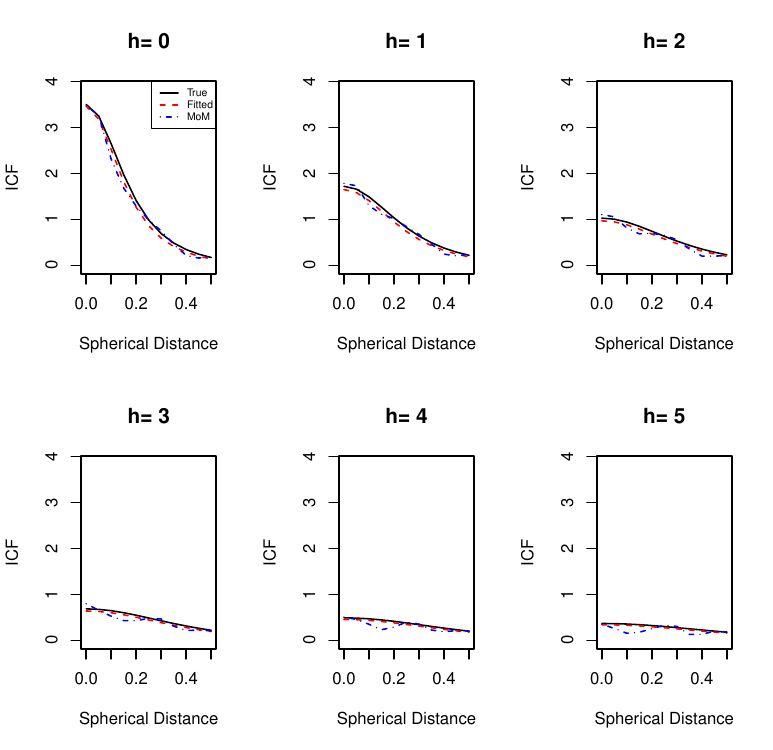}
  \caption{MoM estimates (blue dot-dash line), fitted covariance values (red dashed line), and true covariance values (solid black line) for $IRF(1,1)$ with $\alpha=0.8$, $\beta=0.1$, $\gamma_0=1$, and $\gamma_1=\frac{1}{2 \sqrt{\pi}}$}
  \label{fig:splot}
\end{figure}

\begin{figure}[h!]
\centering
  \includegraphics [width=13cm, height=11cm] {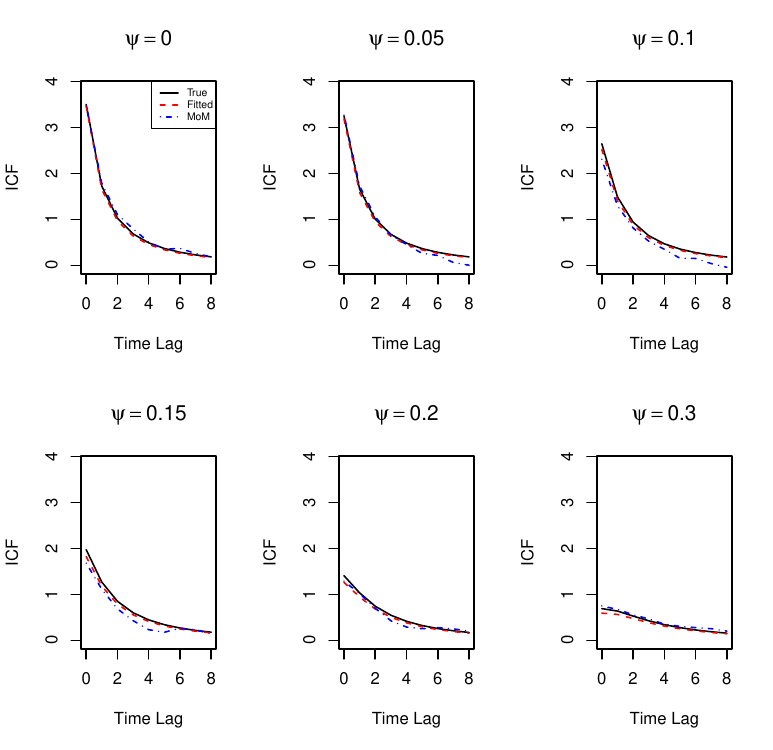}
  \caption{MoM estimates (blue dot-dash line), fitted covariance values (red dashed line), and true covariance values (solid black line) for $IRF(1,1)$ with $\alpha=0.8$, $\beta=0.1$, $\gamma_0=1$, and $\gamma_1=\frac{1}{2\sqrt{\pi}}$}
  \label{fig:tplot}
\end{figure}

\bigskip

In order to access the robustness of the parameter estimates, we also conduct the same simulation multiple times. Figure \ref{fig:para_est_dist} shows the distributions of the 500 parameter estimates with true parameter values $\alpha=0.8$, $\beta=0.1$, and $\gamma_0=1$. "Avg" denotes the average value of the 500 estimates, while "True" refers to the true parameter values used in the simulation dataset. By examining the shape of the distributions and the average values, we can assess the robustness of the estimation.

The performance of parameter estimation can vary based on the true parameter values, as these values influence the decay rates of the covariance functions or intrinsic covariance functions. Table \ref{tab:sim_result} represents the average values and standard deviations of the parameter estimates of the intrinsic covariance function of 500 simulation runs. We conclude that the estimates of $\alpha$ are relatively robust because they are very similar to the true values of the parameters with small variances. Estimation of $\beta$ s also works well, unless its true values are so large that the covariance function decays too fast to be accurately estimated since the parameter $\beta$ is located in the exponential part of $a_\ell(h) = \alpha^\ell e^{-\beta \ell |h|}$.If $\beta$ is too large, the covariance values of pairs will shrink so quickly that they become zero after a time lag of $h \geq 1$. For example, when $\alpha=0.8$ and $\beta=0.9$ as shown in Figure \ref{fig:preliminary_parameters} (bottom left), the covariance value is already almost zero for $h\ge1$. As a result, values of very big $\beta$s become indistinguishable in that their covariance values are always zero unless their time lag $h=0$. The estimation of the scale parameter $\gamma_0$ also performs correctly depending on the other two parameter values. In general, the results of the simulation study allow us to verify the validity of the suggested parameter estimation method.

  \begin{figure}[h!]
\centering
  \includegraphics [width=13cm, height=6cm] {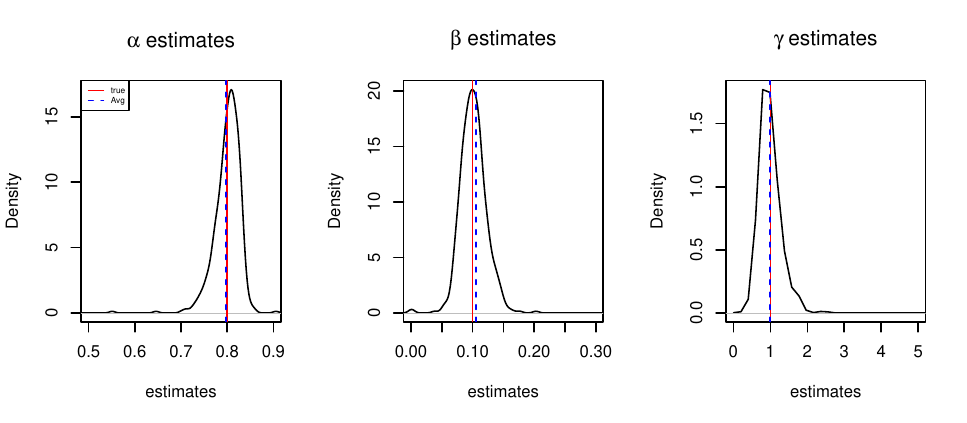}
  \caption{Distributions of estimates with $\alpha=0.8$, $\beta=0.1$, $\gamma_0=1$, and $\gamma_1=\frac{1}{2\sqrt{\pi}}$ for the model of generating functions of Legendre polynomials with $a_\ell(h) = \alpha^\ell e^{-\beta \ell |h|}$ in Table \ref{tab:sim_result}. ``true''(solid, red) means the true parameter value. ``Avg''(blue, dotted) means the average of the estimates. }
  \label{fig:para_est_dist}
\end{figure}

\bigskip

\begin{table}[h!]
\centering
\caption{\label{tab:sim_result}  Parameter estimates of 500 simulations with true parameter values for $IRF(1,1)$ when $\gamma_{1}=\frac{1}{2\sqrt{\pi}}$. Each simulation includes 1,000 locations and 20 temporal points. $a_\ell(h) = \alpha^\ell e^{-\beta \ell |h|}$.}
\begin{tabular}{ |p{1cm}|p{1cm}|p{1cm}||p{1.3cm}|p{1.3cm}|p{1.3cm}||p{1.3cm}|p{1.3cm}|p{1.3cm}|}
 \hline
 \multicolumn{9}{|c|}{Simulation Result} \\
 \hline
 $\alpha$ & $\beta$ & $\gamma_0$ & $Avg(\hat{\alpha})$ & $Avg(\hat{\beta})$  & $Avg(\hat{\gamma_0}$)& $sd(\hat{\alpha})$ & $sd(\hat{\beta})$  & $sd(\hat{\gamma_0}$)\\
 \hline
0.90 & 0.05 & 1.00 & 0.8975 & 0.0124 & 1.0112 & 0.0242 & 0.0517 & 0.5529\\
0.80 & 0.10 & 1.00 & 0.7970 & 0.1240 & 1.0126 & 0.0261 & 0.4451 & 0.3069\\
0.70 & 0.20 & 1.00 & 0.6956 & 0.2067 & 1.0180 & 0.0444 & 0.0603 & 0.4031\\
0.60 & 0.35 & 1.00 & 0.5976 & 0.3676 & 1.0146 & 0.0279 & 0.0989 & 0.2428\\
0.20 & 0.10 & 1.00 & 0.2161 & 0.1368 & 1.4427 & 0.0938 & 0.0812 & 1.5067\\
 \hline
\end{tabular}
\label{tab:sim_result}
\end{table}

\subsection{Analysis of Temperature Anomaly}\label{subsec_temp_anomally}

The National Space Science and Technology Center (NSSTC) at the University of Alabama Huntsville provides data of monthly temperature anomalies  from 1970 to 2019. Each anomaly (measured in degrees Celsius) was computed by subtracting a baseline average from each month's temperature at Earth's surface. The baseline average was measured from 1979 to 1990. Therefore, if some location has a positive temperature anomaly, then it implies that the location at that time point had a higher temperature than the base line average. However, if it shows a negative value, then the temperature of that location was lower than the long-term average. Figure  \ref{fig1} shows examples of January and July in 1980 and 2007 to explain how the data set appears. For the data analysis, 10,368 locations are used with temporal terms of 24 months from 2015 to 2016.\\

\begin{figure}[h!]
\centering
  \includegraphics [width=12cm, height=10cm] {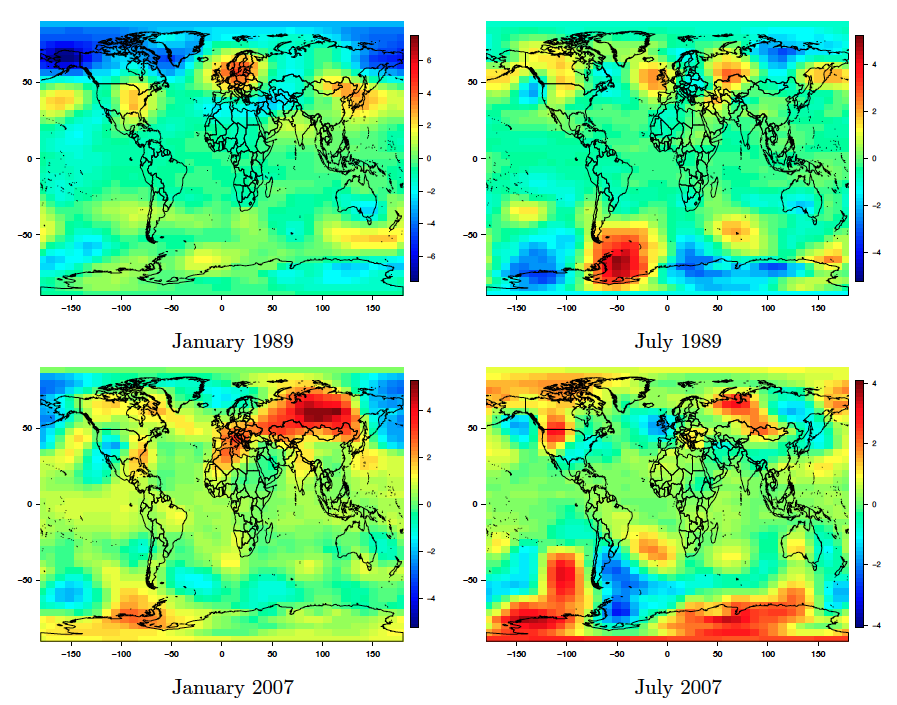}
  \caption{Heat maps for temperature anomaly data for January and July in 1989 and 2007 (Bussberg, 2020) \cite{bussberg2020environmental}}
  \label{fig1}
\end{figure}

\begin{figure}[h!]
\centering
  \includegraphics [width=13cm, height=11cm]
  {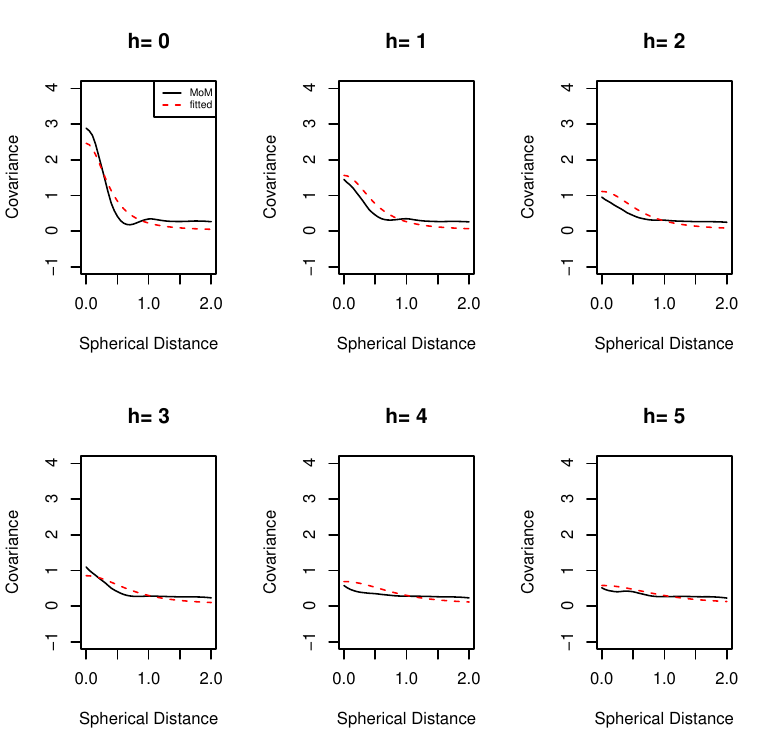}
  \caption{MoM estimates (black solid line) and fitted covariance values (red dashed line) of $IRF(0,0)$ model for the temperature anomaly of 2015-2016. In the first row, time lags (h) are fixed. on the bottom, great circle distances (spherical distances) $\psi$ are fixed.}
  \label{fig_temp_kappa0_inc0_splot}
\end{figure}

\begin{figure}[h!]
\centering
  \includegraphics [width=13cm, height=11cm]
  {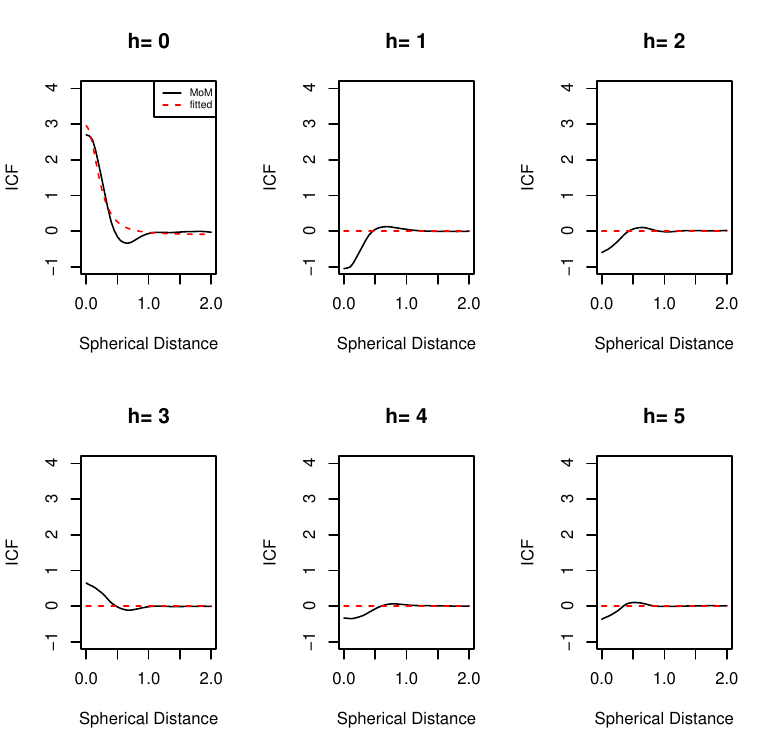}
  \caption{MoM estimates (black solid line) and fitted covariance values (red dashed line) of $IRF(1,1)$ model for the temperature anomaly of 2015-2016. In the first row, time lags (h) are fixed. on the bottom, great circle distances (spherical distances) $\psi$ are fixed.}
  \label{fig_temp_kappa1_inc1_splot}
\end{figure}

\pagebreak

\begin{figure}[h!]
\centering
  \includegraphics [width=13cm, height=11cm]
  {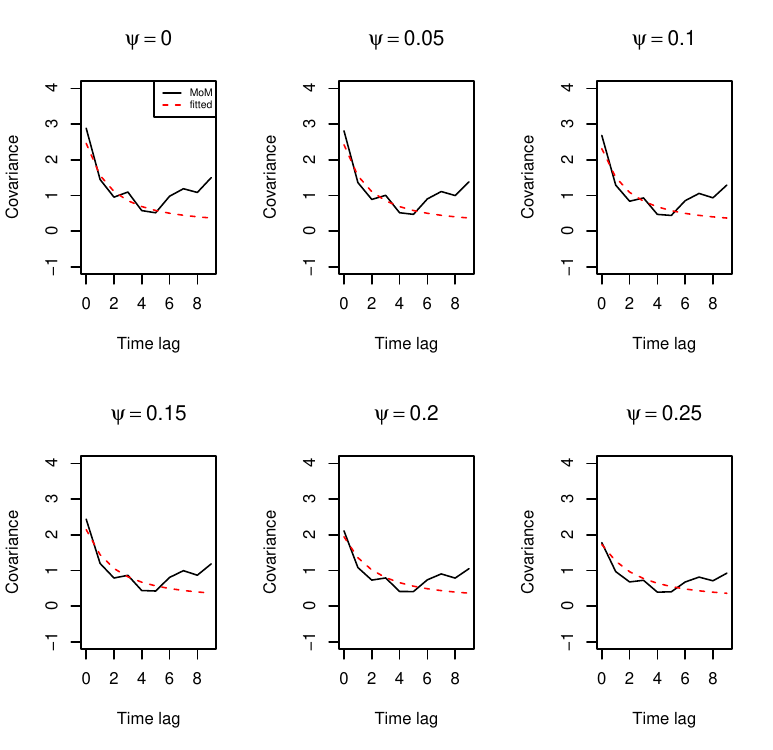}
  \caption{MoM estimates (black solid line) and fitted covariance values (red dashed line) of $IRF(0,0)$ model for the temperature anomaly of 2015-2016. In the first row, time lags (h) are fixed. on the bottom, great circle distances (spherical distances) $\psi$ are fixed.}
   \label{fig_temp_kappa0_inc0_tplot}
\end{figure}

\begin{figure}[h!]
\centering
  \includegraphics [width=13cm, height=11cm]
  {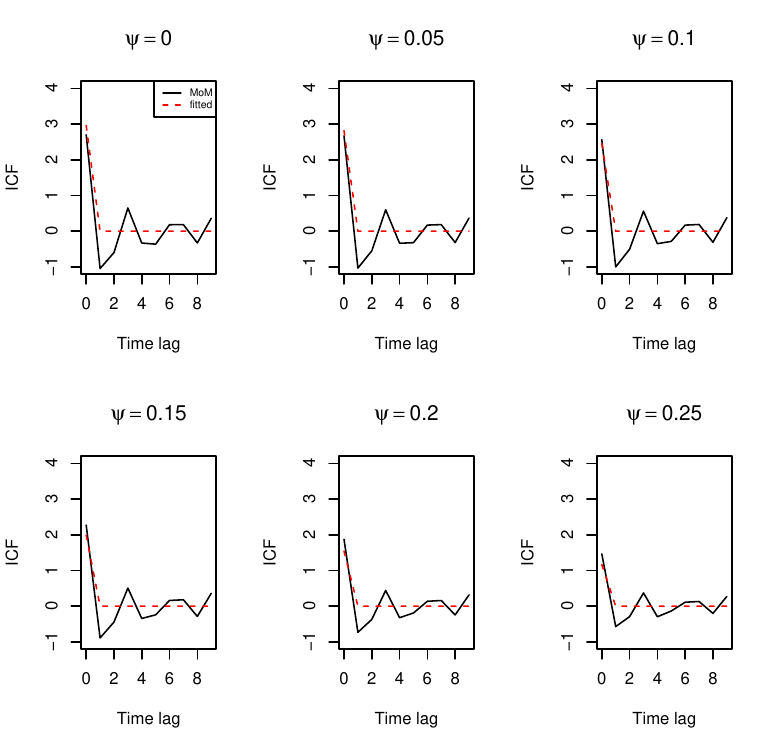}
  \caption{MoM estimates (black solid line) and fitted covariance values (red dashed line) of $IRF(1,1)$ model for the temperature anomaly of 2015-2016. In the first row, time lags (h) are fixed. on the bottom, great circle distances (spherical distances) $\psi$ are fixed.}
   \label{fig_temp_kappa1_inc1_tplot}
\end{figure}

In this analysis, we compare the results obtained using the $IRF(0,0)$ model and the $IRF(1,1)$ model on the dataset. Figures \ref{fig_temp_kappa0_inc0_splot} and \ref{fig_temp_kappa1_inc1_splot} show how the covariance values of $IRF(0,0)$ and the intrinsic covariance function (ICF) values of $IRF(1,1)$ change with great circle distance for a given time lag $h$. As expected, both covariance functions decay as the great circle distance increases. However, the impacts of the time lag $h$ significantly differ between the two. This is further illustrated in Figures \ref{fig_temp_kappa0_inc0_tplot} and \ref{fig_temp_kappa1_inc1_tplot}, which show how the covariance or ICF values change with time lags when the great circle distance is fixed. In Figure \ref{fig_temp_kappa0_inc0_tplot}, we observe that the influence of time lags persists for $IRF(0,0)$ even with relatively large values of $h$. The covariance values of $IRF(0,0)$ decay slowly and do not approach zero, even as time lags increase, as shown in the figure. In time series analysis, graphical methods, such as plotting the covariance or correlation function, are commonly used to detect non-stationarity. If the covariance does not decay rapidly with increasing time lags, it suggests that the process is non-stationary. A characteristic feature of an integrated process is a slowly decaying correlation function (Brockwell, 1991) \cite{Brockwell1991}. In contrast, for $IRF(1,1)$, as shown in Figure \ref{fig_temp_kappa1_inc1_tplot}, the ICF values decay exponentially and approach zero as the time lag increases. This rapid decay may indicate that an integrated process is a suitable model for the dataset.

\bigskip
To choose an appropriate order of non-homogeneity $\kappa$, we can apply the graphical methodology suggested by Bussberg (2020) \cite{bussberg2020environmental}. Given a value of $d$ in $IRF(\kappa,d)$, it has been shown that the intrinsic covariance function, $\phi_{\kappa,d}(\psi, h)$, is spatially homogeneous and temporally stationary. Therefore, for an appropriate estimate, $\hat{\phi}_{\kappa,d}(\overrightarrow{PQ}, \lvert t-s \lvert)$ , one would expect that

\begin{align}
M(n) = \sum_{i=1}^{m_1} \sum_{j=1}^{m_2} \biggl( \hat{\phi}_{n,d}(\psi_i, h_j) - \hat{\phi}_{n+1,d}(\psi_i, h_j) - \{ \hat{\phi}_{n,d}(0, h_j) - \hat{\phi}_{n+1,d}(0, h_j) \} P_k(\cos{\psi_i}) \biggl)^2 \label{kappa_criterion}
\end{align}
to be small for $n \ge \kappa$. The method moment estimator can be used to get $\hat{\phi}_{n,d}(\phi_i, h_j)$. That is,
$$\hat{\phi}_{n,d}(\psi_i, h_j) = \frac{1}{N_{\psi_i, h_j}} \sum_{\{(P,t),(Q,s)\} \in N_{\psi_i, h_j}} \Delta_1^d X_{n,r}(P,t) \Delta_1^d X_{n,r}(Q,s)$$
where $\Delta_1^d X_{n,r}(P,t)$ is an estimator of the truncated process, $\Delta_1^d X_{n}(P,t)$, a residual process of $\Delta_1^d X(P,t)$ regressed on $\{Y_\ell^m(\cdot)\}_{\ell < n}$. The graphical analysis can be derived to estimate the order of non-homogeneity $\kappa$ by using the criterion $M(n)$ in \eqref{kappa_criterion}. Figure \ref{fig:kappa_criterion_plot_log} shows that, for $d=1$, the criterion becomes small when $n \ge 1$, suggesting that $\kappa=1$ is an appropriate level of non-homogeneity. As shown in Figures \ref{fig_temp_kappa1_inc1_splot} and \ref{fig_temp_kappa1_inc1_tplot}, the fitted values and MoM estimates for $IRF(1,1)$ are close to each other. This suggests that, although the true parameter values are unknown, the proposed covariance function performs well from an empirical or data-driven perspective.

\begin{figure}[h!]
\centering
  \includegraphics [width=9cm, height=8cm] {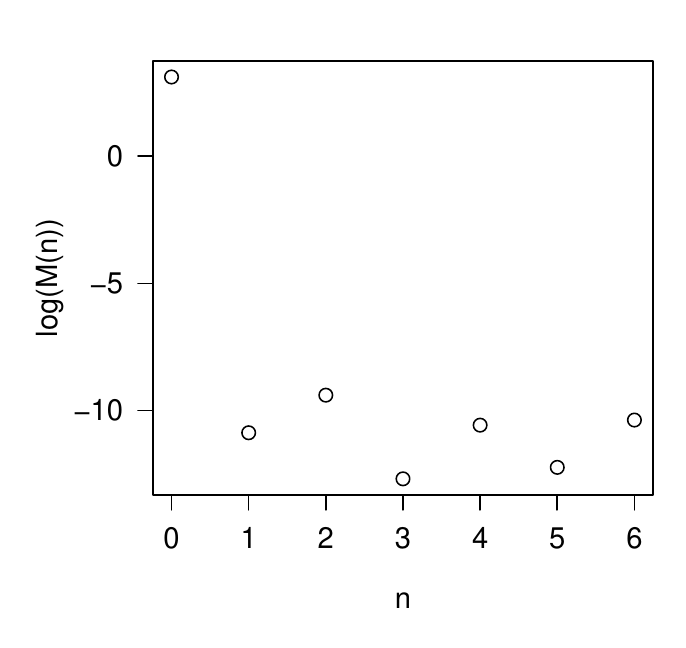}
  \caption{Criterion plot of log M(n) depending on values of n for the temperature anomaly data of 2015-2016.}
   \label{fig:kappa_criterion_plot_log}
\end{figure}

\newpage
2
\pagebreak

\bibliographystyle{unsrt}  
\bibliography{references}

\end{document}